\documentclass[letterpaper, 10 pt, conference]{ieeeconf}

\usepackage[T1]{fontenc}
\usepackage[utf8]{inputenc}
\usepackage{graphicx}
\usepackage{amsmath,systeme,amssymb,amsfonts}
\usepackage[version=4]{mhchem}
\usepackage{siunitx}
\usepackage{longtable,tabularx}
\usepackage{comment}
\usepackage{float}
\usepackage[utf8]{inputenc}
\usepackage{tikz, pgfplots}
\usepackage{textcomp}
\usepackage{tcolorbox} 
\usepackage{graphicx}
\usepackage{amsmath}
\usepackage[version=4]{mhchem}
\usepackage{siunitx}
\usepackage{longtable,tabularx}
\usepackage{tikz}
\usetikzlibrary{positioning, shapes, arrows}
\usepackage{thmtools}
\usepackage{textcomp}
\usepackage{soul}
\usepackage{float,amsfonts,amsthm,color}
\usepackage{amssymb}
\usepackage{graphics} 
\usepackage{booktabs} 
\usepackage{caption, subcaption}
\usepackage{breqn}
\usepackage{array, mathtools}
\usepackage{pgfplots}
\usepackage{siunitx}
\usepackage{algorithm}
\usepackage[noend]{algpseudocode}
\usepackage{tabulary}
\usepackage{stmaryrd} 
\usepackage{varwidth}
\usepackage{comment}
\usetikzlibrary{positioning}

\newtheorem{theorem}{Theorem}
\newtheorem{definition}{Definition}

\newtheorem{proposition}{Proposition}
\newtheorem{remark}{Remark}

\newtheorem{example}{Example}

\newcommand{\R}{\mathbb{R}}

\DeclareMathOperator*{\argmin}{arg\, min} 

\DeclareMathOperator*{\argmax}{arg\, max} 

\IEEEoverridecommandlockouts                              
\title{\LARGE \bf
Passivity, No-Regret, and Convergent Learning in Contractive Games
}

\author{Hassan Abdelraouf$^{1}$, Georgios Piliouras$^{2}$, and Jeff S. Shamma$^{3}$
\thanks{$^{1}$ Department of Aerospace Engineering, University of Illinois at Urbana-Champaign
        {\tt\small hassana4@illinois.edu}}%
\thanks{$^{2}$ Google DeepMind {\tt\small gpil@deepmind.com}
}
\thanks{$^{3}$Department of Industrial and Enterprise Systems Engineering, University of Illinois at Urbana-Champaign
        {\tt\small jshamma@illinois.edu}}%
}

\begin{document}

\maketitle
\thispagestyle{empty}
\pagestyle{empty}

\begin{abstract}
We investigate the interplay between passivity, no-regret, and convergence in contractive games for various learning dynamic models and their higher-order variants. Our setting is continuous time. Building on prior work for replicator dynamics, we show that if learning dynamics satisfy a passivity condition between the payoff vector and the difference between its evolving strategy and any fixed strategy, then it achieves finite regret. We then establish that the passivity condition holds for various learning dynamics and their higher-order variants. Consequentially, the higher-order variants can achieve convergence to Nash equilibrium in cases where their standard order counterparts cannot, while maintaining a finite regret property.  We provide numerical examples to illustrate the lack of finite regret of different evolutionary dynamic models that violate the passivity property. We also examine the fragility of the finite regret property in the case of perturbed learning dynamics. Continuing with passivity, we establish another connection between finite regret and passivity, but with the related equilibrium-independent passivity property. Finally, we present a passivity-based classification of dynamic models according to the various passivity notions they satisfy, namely, incremental passivity, $\delta$-passivity, and equilibrium-independent passivity. This passivity-based classification provides a framework to analyze the convergence of learning dynamic models in contractive games.

\end{abstract}

\section{INTRODUCTION}
No-regret online learning algorithms have become powerful tools in designing adaptive and efficient decision-making strategies within dynamic, uncertain, and competitive environments. These algorithms enable agents to make sequential decisions while minimizing their regret, defined as the difference between the cumulative reward of the algorithm and that of the best fixed action in hindsight \cite{cesa2006prediction}. Among the most well-known no-regret algorithms are \emph{Follow-The-Regularized-Leader} (FTRL) \cite{shalev2007primal}, which includes \emph{Multiplicative-Weights-Update} (MWU)\cite{freund1999adaptive} as a special case, and \emph{Online-Mirror-Descent} (OMD) \cite{hazan2016introduction}. By appropriately selecting a decreasing step size, these algorithms can achieve  $O(\sqrt{T})$ regret.

The intersection of online learning and game theory, often referred to as learning in games,  explores how rational agents learn through repeated interactions in strategic settings, where each agent seeks to maximize its own utility while considering the strategies of others. In this framework, no-regret algorithms allow agents to learn and adapt from past interactions and, for some games, converge to equilibrium strategies without centralized coordination\cite{shoham2008multiagent}. Indeed, fundamental connections have been established between no-regret learning and game-theoretic solution concepts \cite{freund1999adaptive,hart2000simple,foster1997calibrated}.

Passivity is a fundamental input-output property of dynamical systems that abstracts the principles of energy conservation and dissipation in mechanical and electrical systems \cite{willems1972dissipative}. In the online learning context, learning algorithms can be viewed as an input-output operator where the input is the stream of payoffs and the output is the stream of strategies. This perspective enables the application of passivity theory to both analyze and design learning algorithms. For instance, \cite{fox2013population} introduced the notion of $\delta$-passivity to study the convergence of learning dynamics to Nash equilibria in contractive games, showing that the implementation of a learning dynamic model in a game can be viewed as a feedback interconnection between the learning operator (with payoff as input and strategy as output) and the game operator (with strategy as input and payoff as output). Moreover, as demonstrated in \cite{mabrok2016passivity}, if a learning dynamic model fails to satisfy a suitable passivity property, it is possible to construct a higher-order game that results in instability. Recent works \cite{gao2020passivity,gao2023second} have further employed equilibrium-independent passivity to analyze convergence in zero-sum games and to design higher-order variants of learning dynamic models that converge in games where their standard counterparts do not. In this work, we adopt an alternative notion of passivity—specifically, passivity from the payoff vector to the difference between the output strategy and any fixed strategy—which is directly linked to the no-regret property of learning dynamics \cite{cheung2021online}.

The main contributions of this paper are summarized as: 

\begin{itemize}
    \item We show that the strategic higher-order variants of learning dynamic models that have finite regret not only improve convergence in zero-sum games, where their standard counterparts fail to converge, but also preserve the finite regret property (Section \ref{sec: Passivity and No-regret}). 
    \item Numerical examples are provided, illustrating that several evolutionary dynamic models used in population games lack the finite regret property (Section \ref{sec: Passivity and No-regret}).
    \item The finite regret property of a learning dynamic model is shown to be fragile under payoff perturbations (Section \ref{sec: Finite regret fragility}).
    \item A connection is established between finite regret and equilibrium-independent passivity (Section \ref{sec: finite regret and EI-passivity}).
    \item Incremental passivity is shown to be a generalized notion of passivity that implies both equilibrium-independent passivity and $\delta$-passivity (Section \ref{sec: Passivity-based classification}).
    \item We utilize these results, along with existing $\delta$-passivity results for various evolutionary dynamic models \cite{fox2013population,park2018passivity}, to construct a passivity-based classification of learning dynamic models according to the passivity notions they satisfy: incremental passivity, $\delta$-passivity, and equilibrium-independent passivity (Section \ref{sec: Passivity-based classification}).
    \item We demonstrate that learning dynamic models that have finite regret converge globally to the unique Nash equilibrium in strictly contractive games. Moreover, our passivity-based classification provides a comprehensive framework for analyzing the convergence of learning dynamic models in contractive games (Section \ref{sec: Convergence in contracive games}). 
\end{itemize}

\section{NOTATION}
  We denote by $\R_+$ the set of nonnegative real numbers. For a given vector $x \in \R^n$, $x_i$ denotes its $i$-th entry. $[x]_+$ denotes the nonnegative part of $x$, such that its $i$-th component is given by $[x_i]_+ = \max(x_i,0)$. $\text{diag}(x)$ denotes $n \times n$ diagonal matrix whose $i$-th diagonal entry is $x_i$ for all $i=1,\dots,n$. $I_n \in \R^{n\times n} $ is the identity matrix. $\mathbf{1}_n \in \R^n$ is the vector of all ones and $\mathbf{0}_n$ is the vector of all zeros. For two vectors $x,y \in \R^n$,  the inner product is defined as $\langle x,y\rangle = x^T y$ and the Euclidean norm is denoted by $\|x\|_2$. 

 The probability simplex in $\R^n$ is denoted by $\Delta_n$ and is defined as 
\(
    \{s \in \R^n : s_i\geq 0 \; \forall \;  i=1,\dots,n \text{ and } \mathbf{1}^T s =1\}. 
    \)
    $\text{Int}(\Delta_n)$ denotes the interior of the probability simplex. $e_i \in \Delta_n$ denotes the $i$-th vertex of the simplex $\Delta_n$, i.e., the vector whose $i$-th component equals $1$ and all other components equal $0$. The tangent space of $\Delta_n$ is denoted by $T\Delta_n$ where $T\Delta_n = \{y \in \R^n: \mathbf{1}_n^T y = 0\}$. The normal cone to $\Delta_n$ at a point $x\in \Delta_n$ is given by
    \[N\Delta_n(x)= \{y \in \R^n :  y^T (x'-x) \leq 0 \; \forall \; x' \in \Delta_n\}.\]
    For a closed convex set $C \subset \R^n$, $\Pi_C(x)$ denotes the projection of $x \in \R^n $ onto the set $C$ and is given by 
    \[
    \Pi_C(x) = \argmin_{s \in C} \|x-s\|_2.
    \]
    
    $\sigma: \R^n \to \text{Int}(\Delta_n)$ denotes the softmax function such that $\sigma(x) = \exp(x)/\sum_{i=1}^n\exp(x)_i$. The space of square integrable functions over finite intervals is denoted by $\mathbb{L}_e$. i.e., 
    \[\mathbb{L}_e = \left\{ f: \mathbb{R}_+ \to \mathbb{R}^n: \int_0^T f(t)^T f(t) dt < \infty \;\; \forall \;  T \in \mathbb{R}_+ \right\}.\] 
For $f,g \in \mathbb{L}_e$, the truncated inner produced is defined as $\langle f,g\rangle_T  =\int_0^T f(t)^Tg(t) dt$. 
\section{Background}
\subsection{Online Learning and Regret}
\subsubsection{Continuous-time learning dynamics}
In this paper, we consider continuous-time learning dynamics where the payoff vector $p(\cdot): \R_+ \to \R^n$ is assumed to be a continuously differentiable function of time.  In continuous-time  (FTRL) dynamics, the process starts at time $t=0$  and for any $t>0$, the learner computes the cumulative payoff as  $z(t)=z(0)+\int_0^t p(\tau) d\tau$.  The mixed strategy $x(t)\in \Delta_n$ is then computed as $x(t)= C(z(t))$ where 
\begin{equation*}
    C(z) = \argmax_{y\in \Delta_n} y^T z - h(y), 
\end{equation*}
and $h(y)$ is a differentiable strongly convex function. As described in \cite{mertikopoulos2016learning} and \cite{mertikopoulos2017convergence}, this process is governed  by 
 \begin{equation}
     \dot{z} =p,  \; \; \; \; x= C(z),  
     \label{eq: FTRL dynamics}
 \end{equation}
which is known as the continuous-time (FTRL) dynamics. For instance, when $h(y)= \sum_i^n y_i \log(y_i)$, the strategy $x$ becomes $x= \sigma(z)$, and the learning dynamics takes the form
 \begin{equation*}
     \dot{z} = p, \; \; \; x = \frac{\exp(z)}{\sum_{i=1}^n \exp(z)_i}.
 \end{equation*}
Differentiating $x$ w.r.t. time and substituting   yields 
\begin{equation}
    \dot{x}_i  = x_i(p_i-x^Tp),
    \label{eq: RD equation}
\end{equation}
which describes the replicator dynamics (RD) \cite{mertikopoulos2010emergence}, a fundamental model in evolutionary game theory whose long-term rational behavior  was extensively studied in the game-theoretic context in \cite{hofbauer2009time}. 

The continuous-time version of the \emph{Projected-Gradient-Ascent} (PGA) \cite{shalev2012online}, a special case of (OMD), is given by 
\begin{equation}
    \dot{x} = \Pi_{T\Delta_n(x)}(p), 
    \label{eq:DPD}
\end{equation}
where $\Pi_{T\Delta_n(x)}(p)$ denotes the projection of the payoff vector $p(t)$ onto the tangent cone of $\Delta_n$ at $x(t)$. This formulation ensures that $x(t)$ remains in $\Delta_n$ because the velocity $\dot{x}(t)$ is restricted to the feasible directions defined by the tangent cone $T\Delta_n(x(t))$. In the context of evolutionary game theory, this dynamic model known as Direct Projection (DP) dynamics \cite{lahkar2008projection}.

\subsubsection{Regret} In the continuous-time framework, the regret associated with a learning dynamic model is defined as
\begin{equation*}
R_T(\bar{x}) = \int_0^T p(t)^T(\bar{x}-x(t)) dt 
\end{equation*}
which is the difference between the cumulative reward obtained by a fixed strategy $\bar{x} \in \Delta_n$ and that achieved by the dynamics $x(t)$ over the interval $[0,T]$. A learning dynamic model is said to have \emph{no-regret} if 
\begin{equation*}
    \sup_{\bar{x}\in \Delta_n} R_T(\bar{x}) \leq o(t),
\end{equation*}
which means that the time average regret $R_T(\bar{x})/T$ vanishes as $T \to \infty$ for every $\bar{x} \in \Delta_n$. Moreover, the  model is said to have  \emph{finite regret} if, for any payoff trajectory $p(\cdot)$, the regret  at any $T>0$ is bounded from above by a constant for all $\bar{x} \in \Delta_n$. 
\subsection{Passivity}
Passivity theory provides a robust framework for analyzing the stability of feedback interconnections \cite{lozano2013dissipative, van2000l2}. Specifically, connecting two passive systems in a negative feedback loop ensures the stability of the resulting closed-loop system. In this section, we introduce various notions of passivity for dynamical systems, presented in both input-output operator frameworks and state-space representations.
\subsubsection{Input-output operator} \label{input-output operator passivity}
The dynamical system can be represented as an input-output operator $H: \mathbb{U} \to \mathbb{Y}$ where $\mathbb{U}, \mathbb{Y} \subset \mathbb{L}_e$. The input-output operator is:  
\begin{itemize}
    \item Passive if there exists a constant $\alpha \in \R$, such that 
    \begin{equation*}
    \langle Hu, u \rangle_T \ge \alpha, \qquad  \forall u \in \mathbb{U}, T \in \mathbb{R}_+. 
\end{equation*}
\item  $\delta$-passive if there exists a constant $\beta \in \R$ such that 
\begin{equation*}
\langle \dot{(Hu)},\dot{u}\rangle_T \geq \beta \qquad \forall u\in \mathbb{U}, T \in \mathbb{R}_+
\end{equation*}
\item Equilibrium-Independent passive (EI-passive) if for every equilibrium  $(u^*, Hu^*)$, there exists a constant $\gamma \in \R$ such that 
\begin{equation*}
    \langle Hu-Hu^*, u-u^*\rangle_T \geq \gamma \qquad \forall u \in \mathbb{U}, T \in \mathbb{R}_+ 
\end{equation*}
\item Incrementally passive if
\begin{equation*}
    \langle Hu-H\tilde{u}, u-\tilde{u} \rangle_T \geq 0 \qquad \forall u,\tilde{u} \in \mathbb{U}, T \in \mathbb{R}_+ 
\end{equation*}
\end{itemize}
Note that if \(H\) is incrementally passive, then by setting \(\tilde{u}=u^*\) (where \(u^*\) is an equilibrium point), we obtain 
\(
\langle Hu-Hu^*,\, u-u^*\rangle_T \ge 0,\quad \forall\, u\in\mathbb{U},
\)
which implies that \(H\) is EI-passive. 
\begin{remark}
If \(-H\) is \(\delta\)-passive, EI-passive, or incrementally passive, we refer to \(H\) as anti-\(\delta\)-passive, anti-EI-passive, or anti-incrementally passive, respectively.
\end{remark}
\subsubsection{State-space representation}
A dynamical system is represented in the state–space form as

\begin{equation}
    \begin{aligned}
    \dot{x} &= f(x,u), \qquad x(0)=x_0\\
    y & = g(x,u)
    \label{eq: state space model}
\end{aligned}
\end{equation}
where \(x\in\mathbb{R}^n\) is the state vector, \(u(t)\in\mathbb{R}^m\) is the input vector, \(y(t)\in\mathbb{R}^m\) is the output vector , and $x_0 \in \mathcal{M}$ is the initial condition. Let $(u^*,x^*,y^*)$ be an equilibrium condition, i.e., $f(x^*,u^*)=0$ and $y^*=h(x^*,u^*)$. The system \eqref{eq: state space model} is: 
\begin{itemize}
    \item Passive if there exists a continuously differentiable (i.e., $C^1$) storage function $V:\R^n \to \mathbb{R}_+$ such that for all $x_0 \in \mathcal{M}$, and $ T \in \mathbb{R}_+$, 
\begin{equation*}
   V(x(T))-V(x(0))\leq \int_0^T u(t)^T y(t) dt,
\end{equation*}
or equivalently,
\begin{equation*}
\dot{V}:= \nabla V(x) f(x,u) \leq u^T y. 
\end{equation*}
\item \(\delta\)-passive if there exists a $C^1$ storage function \(V_\delta:\mathbb{R}^n \times \mathbb{R}^n\to\mathbb{R}_+\) such that
\[
\dot{V}_\delta(x,\dot{x}) \le \dot{u}^\top \dot{y}.
\]
\item  EI-passive if, for any equilibrium point $(u^*,x^*,y^*)$, there exists a $C^1$ storage function $V_{x^*}(x): \mathbb{R}^n \to \R_+$ with $V_{x^*}(x^*)=0$ and for all $u\in \mathbb{R}^m$,
\begin{equation*}
    \dot{V}_{x^*}(x) \leq (u-u^*)^T(y-y^*).
\end{equation*}
\item Incrementally passive if,  for any two trajectories $(u,x,y)$ and $(\tilde{u},\tilde{x},\tilde{y})$, there exists a $C^1$ storage function $V_{\Delta}(x,\tilde{x}): \mathbb{R}^n \times \R^n \to \mathbb{R}_+$ satisfying  
\begin{equation*}
    \dot{V}_{\Delta}(x,\tilde{x}) \leq (u - \tilde{u})^T (y-\tilde{y}). 
\end{equation*}
\end{itemize}

In particular, if we choose $\tilde{u}(t)=u^*$, $\tilde{x}(0)=x^*$ ( so that $\tilde{y}(t)=y^*$), incremental passivity implies EI-passivity with a storage function $V_{x^*}
(x)=V_{\Delta}(x,x^*)$.
\begin{remark}
    For the linear-time-invariant (LTI) system \(\dot{x} = Ax+Bu, \; y = Cx+Du\) where $x\in \R^n$ and $u,y \in \R^m$, passivity, EI-passivity, $\delta$-passivity, and incremental passivity are equivalent. Moreover, this system is characterized by the transfer function matrix $H(s)=C(sI_n-A)^{-1}B + D$, where $H(s) \in \R^{m \times m}$ is said to be passive if the Hermitian matrix $\hat{H}(jw)+\hat{H}(jw)^*$ is positive semi-definite for all $\omega \in \R$.   
\end{remark}
\subsection{Population Games}
 We consider single population games. Although our results extend to multi-population games, the single population framework simplifies the notations. The payoff function \(F:\Delta_n \to \mathbb{R}^n\) defines the population game by assigning a payoff vector to each strategy \(x\in\Delta_n\). We adopt the following definition of Nash equilibrium for the population game \(F\).

\begin{definition}[Nash Equilibrium]
An element \(x\in\Delta_n\) is a Nash equilibrium for the population game \(F\) if
\[
x^\top F(x) \ge z^\top F(x) \quad \forall z\in\Delta_n.
\]
\end{definition}

A population game may have multiple Nash equilibria.  We denote the set of all Nash equilibria by \(\mathbb{NE}(F)\). In this paper, we focus on \emph{contractive} population games.

\begin{definition}[Contractive Games \cite{hofbauer2009stable}]\label{def: contracrive games}
A population game \(F\) is said to be \emph{contractive} if 
\[
(x-y)^\top \bigl(F(x)-F(y)\bigr) \le 0 \quad \forall  x,y\in\Delta_n.
\]
If equality holds if and only if \(x=y\), then \(F\) is called strictly contractive.
\end{definition}
In a contractive game, the set \(\mathbb{NE}(F)\) is convex; moreover, if \(F\) is strictly contractive, then it has a unique Nash equilibrium \cite{sandholm2015population}. 

\section{Passivity and No-regret} \label{sec: Passivity and No-regret}
In this section, we analyze the passivity properties of the continuous-time versions of (FTRL) dynamics \eqref{eq: FTRL dynamics},  (DP) dynamics \eqref{eq:DPD}, and their strategic higher-order variants. The finite regret property of a learning dynamic model is related to its passivity.  Specifically,  if the learning dynamic model is passive from $p$ to $x-\bar{x}$ for every fixed strategy $\bar{x} \in \Delta_n$, then the model achieves finite regret. 
\begin{theorem}[Passivity and No-regret \cite{cheung2021online}]\label{thrm:no_regret_theorem} 
    Any passive learning dynamics from $p$ to $x-\bar{x}$ for every $\bar{x} \in \Delta_n$ has finite regret. 
\end{theorem}
The proof is based on the fact that the passivity from $p$ to $x-\bar{x}$ for every $\bar{x} \in \Delta_n$ implies the existence of a positive semi-definite storage function $V(\xi)\geq 0$, where $\xi(t)$ represents the state vector of the learning dynamics, such that for all $T>0$  
\begin{equation*}
       \int_0^T p^T(x-\bar{x})dt \ge V(\xi(T))-V(\xi(0)). 
   \end{equation*}
   Hence, the regret for not using any fixed action $\bar{x} \in \Delta_n$ is 
   \begin{equation*}
       R_T(\bar{x})=\int_0^T p^T (\bar{x}-x) dt \leq V(\xi(0))-V(\xi(T)).
   \end{equation*}
   Since $V(\xi(T))\ge 0$ for all $T>0$, the above inequality simplifies to $R_T(\bar{x}) \leq V(\xi(0))$,  which shows that the regret $R_T(\bar{x})$ is bounded above by $V(\xi(0))$.  

Previous work \cite{cheung2021online} has shown that the continuous-time (FTRL) dynamics \eqref{eq: FTRL dynamics} is passive from $p$ to $x-\bar{x}$ for every $\bar{x} \in \Delta_n$.  Similarly, replicator dynamics (RD) \eqref{eq: RD equation}, a special case of (FTRL) dynamics,  is passive lossless from $p$ to $x-\bar{x}$ for every $\bar{x} \in \Delta_n$ \cite{mabrok2016passivity}.  The following theorem shows that the (DP) dynamics \eqref{eq:DPD} forms a passive mapping from $p$ to $x - \bar{x}$ for every $\bar{x} \in \Delta_n$ and hence has finite regret. 
 \begin{theorem} \label{thrm: DP finite regret}
        Direct Projection (DP) dynamics (\ref{eq:DPD}) has finite regret w.r.t. any $\bar{x} \in \Delta_n$. 
    \end{theorem}
    \begin{proof}
     To apply Theorem \ref{thrm:no_regret_theorem}, we need to show that (DP) dynamics is passive from $p$ to $x-\bar{x}$ for every $\bar{x} \in \Delta_n$. For an arbitrary $\bar{x} \in \Delta_n$, consider the storage function
      \begin{equation*}
          V(x)= \frac{1}{2} \|x-\bar{x}\|_2^2.
      \end{equation*}
      Clearly, $V(x)\ge 0$ and $V(\bar{x})=0$. The derivative of $V(x)$ is
      \begin{equation*}
          \begin{aligned}
              \frac{d}{dt} V(x) &= (x-\bar{x})^T \dot{x} = (x-\bar{x})^T \Pi_{T\Delta_n(x)}(p) \\
              &= (x-\bar{x})^T(p-\Pi_{N\Delta_n(x)}(p)) \\
              &= (x-\bar{x})^T p + (\bar{x}-x)^T \Pi_{N\Delta_n(x)}(p)\\
              &\leq (x-\bar{x})^T p, 
          \end{aligned}
      \end{equation*}
where the last inequality follows because $\Pi_{N\Delta_n(x)}(p)$ lies in the normal cone of $\Delta_n$ at $x$, then $(\bar{x}-x)^T \Pi_{N\Delta_n(x)}(p) \leq 0$ for all $\bar{x} \in \Delta_n$. 
Hence, (DP) dynamics is a passive from $p$ to $(x-\bar{x})$ for every $\bar{x} \in \Delta_n$. By Theorem  \ref{thrm:no_regret_theorem}, it follows that (DP) dynamics has finite regret w.r.t. every $\bar{x} \in \Delta_n$. 
    \end{proof}
Furthermore, for a learning dynamic model, establishing finite regret does not require checking finite regret w.r.t. every point $\bar{x} \in \Delta_n$. In fact, it is sufficient to verify finite regret w.r.t. the vertices of the simplex. This is formally stated in the following proposition. 
\begin{proposition}
    A learning dynamic model has a finite regret w.r.t. any point $\bar{x} \in \Delta_n$ (excluding the vertices) if and only if it has a finite regret w.r.t. all vertices of $\Delta_n$.
\end{proposition}
\begin{proof}
    ($\Rightarrow$)Assume that the learning dynamic model has a finite regret w.r.t. all the vertices of $\Delta_n$.  That is, for each vertex $e_i$ of $\Delta_n$, there exists a constant $\beta>0$ such that $\int_0^T p^T (e_i-x)dt \leq \beta$ for all $T>0$. Because $\Delta_n$ is convex, any point $\bar{x}\in \Delta_n$ can be expressed as $\bar{x} = \sum_{i=1}^n a_i e_i$, where $\sum_{i=1}^n a_i =1$ and $a_i\geq 0$ for all $i=1, \dots,n$.  Then, the regret w.r.t any point $\bar{x} \in \Delta_n$  over $T>0$ is 
    \begin{equation*}
    \begin{aligned}
    R_T(\bar{x}) &= \int_0^T p^T(\bar{x}-x)dt = \int_0^T p^T(\sum_{i=1}^n a_i e_i -x) dt \\
       &= \sum_{i=1}^n a_i\int_0^T p^T(e_i-x) dt \leq \sum_{i=1}^n a_i \beta = \beta. 
    \end{aligned}
    \end{equation*}
    Thus the model has finite regret w.r.t. any point $\bar{x} \in \Delta_n$. 
    ($\Leftarrow$)Conversely, assume that the learning dynamic model has finite regret w.r.t. every $\bar{x}\in \Delta_n$ (excluding the vertices). That is, there exists $\beta >0$ such that  $\int_0^T p^T(\bar{x}-x) dt\leq \beta$ for each $\bar{x} \in \Delta_n$ (excluding the vertices) and all $T>0$. For every vertex $e_i$ of $\Delta_n$, choose a sequence $\{\bar{x}_k\} \subset \Delta_n$ with $\bar{x}_k \to e_i$ as $k\to \infty$. By the integral continuity, it follows that, $\lim_{k\to \infty} \int_0^T p^T(\bar{x}_k-x) dt = \int_0^T p^T(e_i-x) dt \leq \beta$  for all $T>0$. Hence, the regret w.r.t. each vertex is finite. 
\end{proof}

However (FTRL) dynamics \eqref{eq: FTRL dynamics} and (DP) dynamics \eqref{eq:DPD} have finite regret, they fail to converge in zero-sum games \cite{mertikopoulos2018cycles}. This limitation has motivated the development of a discounted version \cite{gao2020passivity}  and higher-order variants \cite{arslan2006anticipatory}, \cite{laraki2013higher},\cite{gao2023second} of (FTRL) dynamics to ensure convergence in such settings. In the subsequent theorems, we examine the finite regret properties of these modified versions, specifically, the strategic higher-order variants. 

The Strategic-Higher-Order FTRL (SHO-FTRL) dynamics is given by:
\begin{equation}
    \begin{aligned}
    \dot{z} &= p -\gamma(x-\xi), \; \; \; \dot{\xi} = \lambda(x-\xi) \\ 
           x&= C(z). 
\end{aligned}
\label{eq: SHO-FTRL}
\end{equation}
which  can be interpreted as a negative feedback interconnection between the continuous-time  (FTRL) dynamics \eqref{eq: FTRL dynamics} and the linear time invariant (LTI) system with the transfer function matrix $g(s)I_n$ where $g(s)= \gamma s /(s+ \lambda)$. Figure \ref{fig: SHO-FTRL} illustrates this interconnection. 
 \begin{figure}[H]
    \centering
    \includegraphics[scale=0.35]{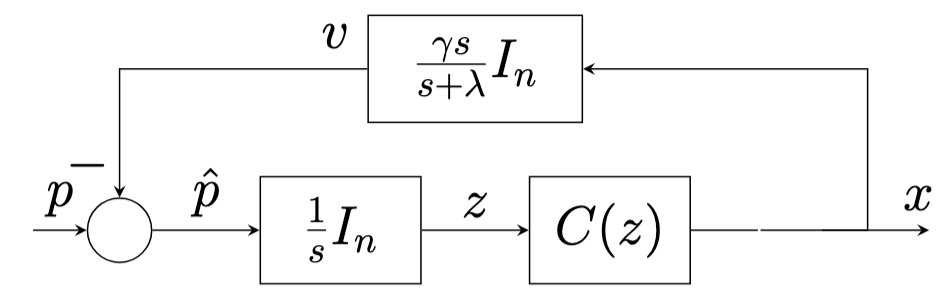}
    \caption{ (SHO-FTRL) dynamic model}
    \label{fig: SHO-FTRL}
\end{figure}

\begin{theorem} \label{thrm:SHO-FTRL no regret}
    (SHO-FTRL) dynamic model \eqref{eq: SHO-FTRL} has finite regret w.r.t. any $\bar{x} \in \Delta_n$. 
\end{theorem}

\begin{proof}
    The LTI system 
    \begin{equation}
        \dot{\xi} = \lambda(x-\xi), \; \; v = \gamma(x-\xi)
        \label{eq: LTI SHO-FTRL}
    \end{equation}
is characterized by the transfer function matrix $G(s)= (\gamma s /(s+\lambda)) I_n$. Given that $\lambda,\gamma>0$, then $\text{Re}(\hat{g}(j\omega))= (\gamma \omega^2/(\lambda^2 + \omega^2)) \geq 0 $ and the hermitian matrix $G(jw)+G(jw)^*$ is positive semi-definite for all $\omega \in \mathbb{R} $. Therefore, the LTI system \eqref{eq: LTI SHO-FTRL} is passive. Since, for LTI systems, passivity is equivalent to EI-passivity, the system \eqref{eq: LTI SHO-FTRL} is EI-passive. Hence, for fixed  $\bar{x}\in \Delta_n$ and the equilibrium point ($\bar{\xi}=\bar{x}$ and $\bar{v}=0$),  there exists a storage function  $V_{\bar{\xi}}(\xi) \geq 0  $ \footnote {$V_{\bar{\xi}}(\xi)=\|\xi-\bar{\xi}\|_2^2/2$}and $V_{\bar{\xi}}(\bar{\xi})=0$ such that $\dot{V}_{\bar{\xi}}(\xi) \leq v^T (x-\bar{x})$. For the feedback interconnection, define the storage function $$V(z,\xi)= \max_{y \in \Delta_n} (z^T y - h(y)) -(z^T \bar{x}-h(\bar{x}))+ V_{\bar{\xi}}(\xi).$$
Clearly, $V(z,\xi)\geq 0$. The derivative of $V(z,\xi)$ is 
\begin{equation*}
    \begin{aligned}
        \frac{d}{dt} V(z,\xi) &= (x-\bar{x})^T \dot{z} + \dot{V}_{\bar{\xi}}(\xi)   \\
        & \leq (x-\bar{x})^T \hat{p} + (x-\bar{x})^T v   \\
        &= (x-x^T)(p-v) + (x-\bar{x})^T v = (x-\bar{x})^T p. 
    \end{aligned}
\end{equation*}
Thus the negative feedback interconnection is passive from $p$ to $x-\bar{x}$ for every $\bar{x} \in \Delta_n$. Applying theorem \ref{thrm:no_regret_theorem} implies that (SHO-FTRL) dynamics has finite regret. 
\end{proof}
Similarly, the Strategic-Higher-Order DP (SHO-DP) dynamics can be written as:  
\begin{equation}
    \begin{aligned}
    \dot{x}&= \Pi_{T\Delta_n(x)}(p-\gamma(x-\xi)) \\
    \dot{\xi}&= \lambda(x-\xi) 
\end{aligned}
\label{eq: SHO-DPD}
\end{equation}
Figure~\ref{fig:SHO-DPD} illustrates the block diagram of the (SHO-DP) dynamic model.

\begin{figure}[H]
    \centering
    \includegraphics[scale=0.35]{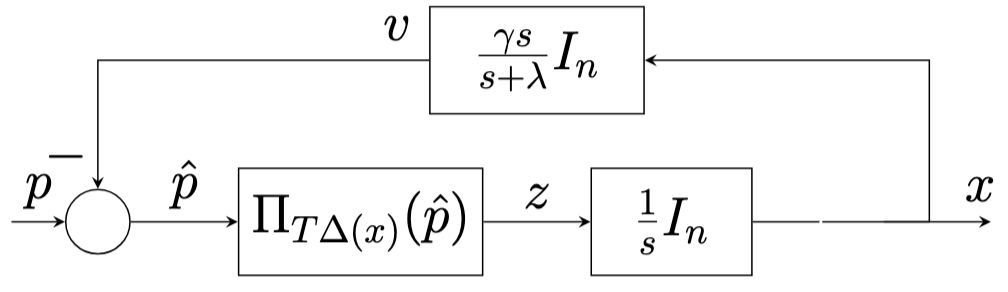}
    \caption{ (SHO-DP) dynamic model}
    \label{fig:SHO-DPD}
\end{figure}
\begin{theorem}
    (SHO-DP) dynamics \eqref{eq: SHO-DPD} has  finite regret w.r.t every $\bar{x} \in \Delta_n$. 
\end{theorem}

\begin{proof}
     Let \(\bar{x} \in \Delta_n\) be arbitrary. Using the EI-passivity of the LTI system  \eqref{eq: LTI SHO-FTRL} and the corresponding storage function $V_{\bar{\xi}}(\xi)$, consider the storage function 
     \begin{equation*}
         V(x,\xi)= \frac{1}{2} \|x-\bar{x}\|_2^2 + V_{\bar{\xi}}(\xi). 
     \end{equation*}
     The time derivative of $V(x,\xi)$ is
     \begin{align*}
         \frac{d}{dt} V(x,\xi) &= (x-\bar{x})^T \dot{x}+\dot{V}_{\bar{\xi}}(\xi) \\
         &\leq (x-\bar{x})^T \Pi_{T\Delta_n(x)}(\hat{p}) + (x-\bar{x})^T v \\
         &= (x-\bar{x})^T (\hat{p}-\Pi_{N\Delta_n(x)}(\hat{p}))  + (x-\bar{x})^T v \\
         & \leq (x-\bar{x})^T (p-v) + (x-\bar{x})^T v \\
         &= (x-\bar{x})^T p. 
     \end{align*}
     Therefore, the (SHO-DP) dynamics is passive  from $p$ to $x-\bar{x}$.  Theorem \ref{thrm:no_regret_theorem} implies that it has finite regret w.r.t. every $\bar{x}\in \Delta_n$. 
\end{proof}
Building on our demonstration that  FTRL, DP, SHO-FTRL, and SHO-DP dynamics have finite regret, we now extend our analysis to the evolutionary dynamic models used in population games. These models describe how the population state, represented by a strategy $x(t) \in \Delta_n$, evolves in response to a payoff vector $p(t) \in \R^{n}$. We consider models formulated in state-space form as
\begin{equation}
    \dot{x}(t) = \mathcal{V}(x(t),p(t))
\end{equation}
 where $\dot{x}(t)$ lies in the tangent space of $\Delta_n$. We assume that $\mathcal{V}(x,p)$ is well defined in the sense that it is continuous in $(x(t),p(t))$ and for every initial condition $x(0) \in \Delta_n$ and every payoff trajectory $p(\cdot)$, there exists a unique solution $x(t)$ for all $t>0$.  Examples of such evolutionary dynamic models include: 
\begin{enumerate}
    \item Replicator dynamics (RD) \eqref{eq: RD equation}.
    \item Direct projection (DP) dynamics \eqref{eq:DPD}.
    \item BNN dynamics \cite{brown1950solutions}:
    \begin{equation}
        \dot{x}_i = [p_i-x^Tp]_+ - x_i \sum_{j=1}^n [p_j-x^Tp]_+
        \label{eq: BNN dynamics}
    \end{equation}
    \item Smith dynamics \cite{smith1984stability}: 
    \begin{equation}
        \dot{x}_i = \sum_{j=1}^n x_j [p_i-p_j]_+ - x_i \sum_{j=1}^n [p_j-p_i]_+
        \label{eq: Smith dynamics}
    \end{equation}
    \item Logit dynamics \cite{hofbauer2007evolution}:
    \begin{equation}
        \dot{x} = \sigma(p) -x 
        \label{eq: logit dynamics}
    \end{equation}
    \item Target Projection (TP) dynamics \cite{tsakas2009target}
     \begin{equation}
    \dot{x} = -x+\Pi_\Delta(x+p) 
    \label{eq: target projection dynamics}
    \end{equation}
    \item Exponential replicator dynamics (Ex-RD)\cite{gao2020passivity}
    \begin{equation}
        \begin{aligned}
            \dot{z} & = \lambda(p-z) \\
            x & = \sigma(z) 
        \end{aligned}
        \label{eq: EX-RD}
    \end{equation}
\end{enumerate}
The following example is used to study the finite regret property of the introduced evolutionary dynamic models above. 
\begin{example} \label{ex: no regret counter example}
    Consider an environment where the payoff vector is given by  $p(t) = \begin{bmatrix}
        \sin(t) & 0.5
    \end{bmatrix}^T$, and the learner's decision at time $t$ is  
     $x(t) = \begin{bmatrix}
        x_1(t) & x_2(t)
    \end{bmatrix}^T$, with $x_1(t)$ and $x_2(t)$ denoting the probability of selecting $p_1(t) =sin(t)$ and $p_2(t)=0.5$ respectively, for all $t \in \mathbb{R}_+$.  
\end{example}
For this environment, the optimal strategy is to choose $p_1(t)$ with probability $1$ if $\sin(t)\geq0.5$ and $p_2(t)$ with probability $1$ if $\sin(t) < 0.5$, i.e., 
\small{
\[
x_{\mathrm{opt}}(t) =
\begin{cases}
\begin{bmatrix} 1 & 0 \end{bmatrix}^T, & \text{if } \sin(t) \ge 0.5, \\[2mm]
\begin{bmatrix} 0 & 1 \end{bmatrix}^T, & \text{if } \sin(t) < 0.5.
\end{cases}
\]}
Under this strategy, the average reward is 
\begin{equation*}
    \lim_{T \to \infty} \frac{1}{T} \int_0^T \max(\sin(t),0.5) dt \approx 0.609
\end{equation*}

 In contrast, If the learner were to consistently choose $p_1(t)$ (i.e. $x(t)= \begin{bmatrix} 1 & 0\end{bmatrix}^T \; \forall t \in \R_+$), the average reward would be 0, and if only $p_2(t)$ were selected (i.e. $x(t) = \begin{bmatrix} 0 & 1\end{bmatrix}^T \; \forall t \in \R_+$), the average reward would be 0.5. More generally, from the convexity of $\Delta_2$, choosing any fixed strategy for all $t \in \R_+$, the learner will get an average reward between $0$ and $0.5$. 

Implementing BNN, Smith, Logit, and (TP) dynamics on the environment introduced in Example \ref{ex: no regret counter example} yields average rewards that converge to approximately \(0.374, 0.453, 0.36, 0.467 \) respectively. These values are lower than the average reward of \(0.5\) that would be achieved by consistently selecting $p_2(t)$ (i.e., \(x(t) = e_2 \; \forall t \in \R_+\)), as shown in Figure \ref{fig:EDMs performacne in Example 1}.  Consequently, $\int_0^T p^T (x-e_2)dt$ diverges to \(-\infty\) as \(T\to\infty\); that is, \(\langle p,x-e_2\rangle_T\) does not have a lower bound.  This indicates that  BNN, Smith, Logit, and (TD) dynamics are not  passive from \(p\) to \(x-e_2\) and therefore,  do not guarantee finite regret for any environment. 
\begin{figure}[H]
    \centering
    \includegraphics[scale=0.35]{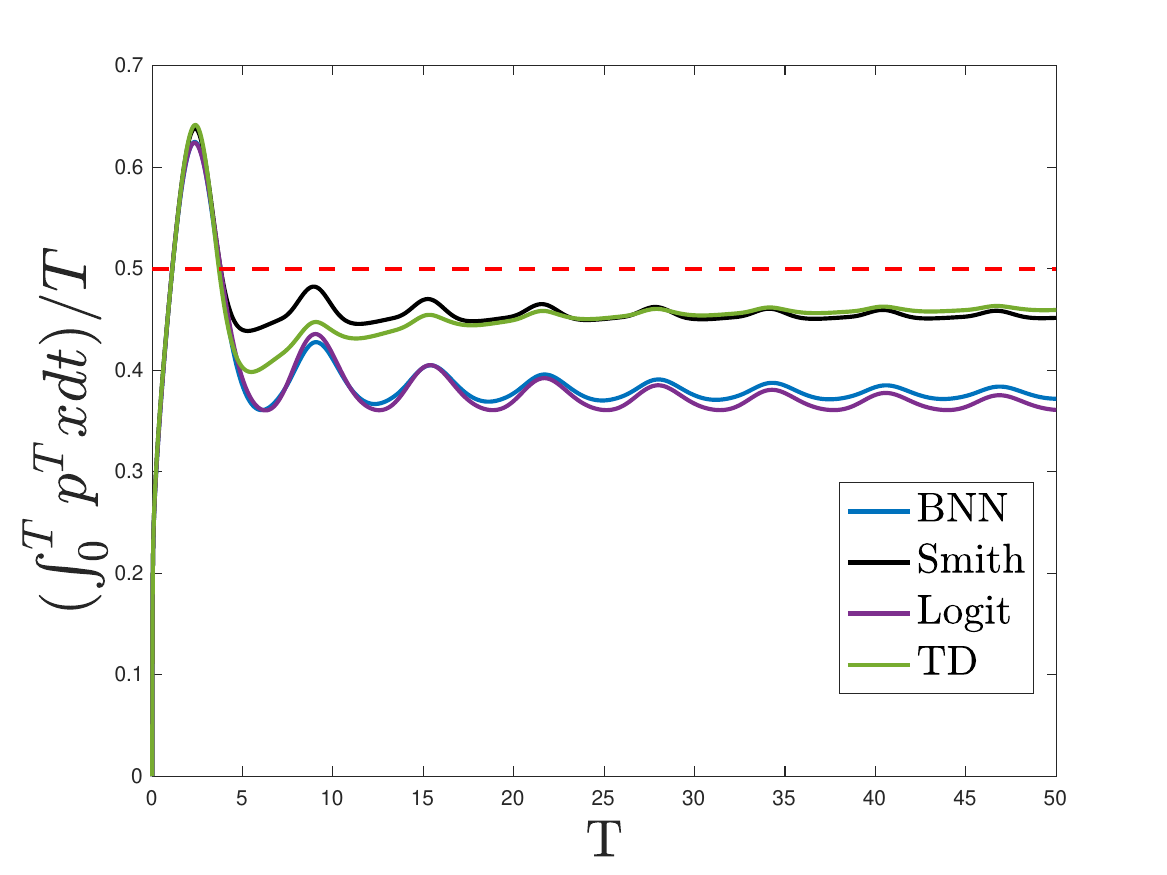}
    \caption{BNN, Smith, Logit, (TD) dynamics performance in the environment introduced in Example \ref{ex: no regret counter example}}
    \label{fig:EDMs performacne in Example 1}
\end{figure}
Consider the case where $p(t) =e_1=\begin{bmatrix}
    1 &0
\end{bmatrix}^T$. In this setting, if the learner always selects $p_1(t)$ (i.e., $x(t) = e_1 \;\forall \; t \in \R_+$ ), the average reward is $1$. However, under (Ex-RD) \eqref{eq: EX-RD}, the strategy $x(t)$ converges to $\sigma(e_1)$ yielding an average reward $e_1 \sigma(e_1) \leq 1$ which indicates that (Ex-RD) \eqref{eq: EX-RD} is not passive from $p$ to $x-e_1$ and hence does not have finite regret. 

\begin{tcolorbox}
 \textbf{Summary of Regret Properties}:
 \begin{itemize}
      \item (RD) \eqref{eq: RD equation}, (DP) \eqref{eq:DPD}, (SHO-FTRL) \eqref{eq: SHO-FTRL}, and (SHO-DPD) \eqref{eq: SHO-DPD} dynamic models have finite regret.
      \item BNN \eqref{eq: BNN dynamics}, Smith \eqref{eq: Smith dynamics}, Logit \eqref{eq: logit dynamics}, (TD) \eqref{eq: target projection dynamics} and (Ex-RD) \eqref{eq: EX-RD} dynamic models \textbf{do not} have finite regret.
    \end{itemize}
\end{tcolorbox}

\section{Finite regret Fragility} \label{sec: Finite regret fragility}
In particular settings, payoff vectors are often subject to delays or dynamic modifications rather than being received instantaneously. This section explores how such modifications affect the finite regret guarantees of standard learning dynamics. As an example, we consider replicator dynamics with latency, where the effective payoff is replaced by a delayed or filtered version of the actual payoff signal. Specifically, the dynamics are given by
\begin{equation}
\begin{aligned}
    \dot{x}_i &= x_i(\hat{p}_i-x^T\hat{p}) \\ 
    \dot{\hat{p}} &= \lambda(p-\hat{p}),
\end{aligned}
\label{eq: RD with latency}
\end{equation}
where $\lambda>0$ and $\hat{p}$ is the delayed version of $p$. 
The following theorem shows that replicator dynamics with latency \eqref{eq: RD with latency} does not have finite regret. 

\begin{theorem}\label{thrm: RD with latency} 
    The replicator dynamics with latency \eqref{eq: RD with latency} does not have finite regret. 
\end{theorem}
\begin{proof}
    First, we have, 
    \begin{equation*}
        \begin{aligned}
            \frac{\dot{x}_i}{x_i} &=  (p_i-\frac{1}{\lambda} \dot{\hat{p}}_i) -x^T (p-\frac{1}{\lambda}\dot{\hat{p}}) \\
        & = (p_i-x^T p) - \frac{1}{\lambda} (\dot{\hat{p}}_i-x^T \dot{\hat{p}})
        \end{aligned}
    \end{equation*}
    Rearranging and integrating both sides over \([0,T]\) gives the regret with respect to the vertex \(e_i\) as 
    \begin{equation*}
    \begin{aligned}
         R_T(e_i) &= \int_0^T (p_i-x^T p) dt = \int_0^T  \frac{\dot{x}_i}{x_i} dt  + \int_0^T  \frac{1}{\lambda} (\dot{\hat{p}}_i-x^T \dot{\hat{p}}) dt\\
         &= \underbrace{\int_0^T \frac{\dot{x}_i}{x_i} dt}_{a(t)} + \frac{1}{\lambda}\int_0^T  \dot{\hat{p}} dt - \frac{1}{\lambda} \int_0^T x^T \dot{\hat{p}} dt. 
    \end{aligned}
    \end{equation*}
    Integration the last term by parts yields,  
    \begin{equation*}
         R_T(e_i) = a(t) + \frac{1}{\lambda} \underbrace{ \left( \int_0^T \dot{\hat{p}} dt - x(t)^T \hat{p}(t) \Big|_0^T\right)}_{b(t)}  + \frac{1}{\lambda} \underbrace{\int_0^T \hat{p}^T\dot{x} dt}_{c(t)} 
    \end{equation*}
\noindent
\emph{Bound on \(a(t)\).} Note that
\(
a(t)= \log\bigl(x_i(T)\bigr)-\log\bigl(x_i(0)\bigr)
\le -\log\bigl(x_i(0)\bigr),
\)
which is bounded above.

\noindent
\emph{Bound on \(b(t)\).} Since we assume \(p\) is bounded, \(\hat{p}\) is also bounded. Hence,
\(
b(t)=\hat{p}_i(T)-\hat{p}_i(0)-x(T)^\top \hat{p}(T)+x(0)^\top \hat{p}(0)
\)
is uniformly bounded.

\noindent
\emph{Analysis of \(c(t)\).} For the last term,
 \begin{equation*}
     \begin{aligned}
         c(t)&=  \int_0^T  \sum_{i=1}^n \hat{p}_i^T \dot{x}_i \; dt \\
         &=  \int_0^T \sum_{i=1}^n \left(\hat{p}_i^T x_i(\hat{p}_i-x^T \hat{p})\right) \; dt  \\
         &=  \int_0^T \left(\sum_{i=1}^n x_i \hat{p}_i^2 - (x^T \hat{p}) \sum_{i=1}^n \hat{p}_ix_i \right) dt \\
         &= \int_0^T {\hat{p}(t)}^T \underbrace{\left(\text{diag}(x(t))-x(t)x(t)^T\right)}_{Q(t)}\hat{p}(t) \; dt 
     \end{aligned}
 \end{equation*}
 The matrix $Q(t)$ is positive semi-definite and has zero eigenvalue associated to the eigenvector $\mathbf{1}_n$  \cite[Proposition~2]{gao2017properties}. Consequently, the integrand is nonnegative and vanishes precisely if \(\hat{p}(t)=c\mathbf{1}_n\) for some constant \(c\in\mathbb{R}\) or \(x(t)\) is a vertex. Hence, trajectories that experience recurring behaviors bounded away from these conditions  accumulate unbounded regret, violating the finite regret property.  
 
\end{proof}

\begin{example} \label{ex: RD_latency}
    Consider the case where $p(t)= \begin{bmatrix}
        \sin(t) & -\sin(t)
    \end{bmatrix}^T $.  Choosing either $p_1(t)$ or $p_2(t)$ for all $t \in \R_+$ with probability 1 results in an average reward of zero. However, when applying the replicator dynamics with latency \eqref{eq: RD with latency}, the average reward is approximately \(-0.106\), which is strictly worse than consistently selecting either \( p_1(t) \) or \( p_2(t) \) for all \( t \in \mathbb{R}_+ \). Figure~\ref{fig: RD with latency performace} displays the average reward (left) and the accumulation of regret w.r.t. \(e_1\) and \(e_2\) (right), illustrating the failure to maintain finite regret.  
\end{example}

\begin{figure}[H]
    \centering
    \includegraphics[scale=0.4]{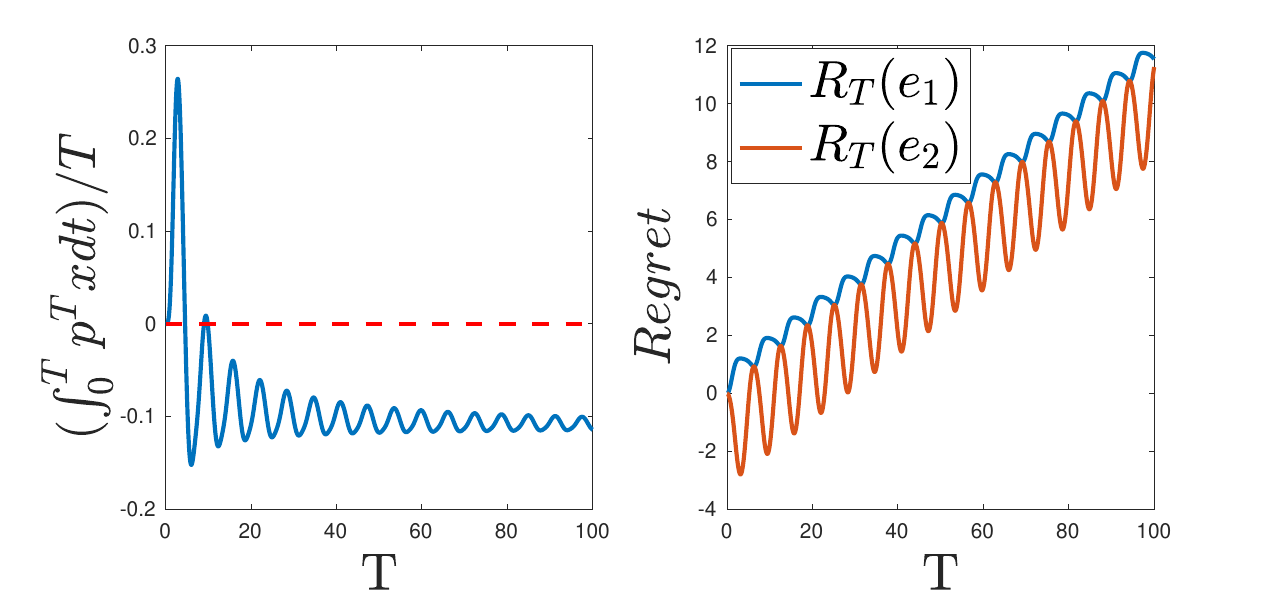}
    \caption{The performance of replicator dynamics with latency  in the environment introduced in Example \ref{ex: RD_latency} }. 
    \label{fig: RD with latency performace}
\end{figure}
\section{Finite Regret and EI-passivity} \label{sec: finite regret and EI-passivity}

EI-passivity  is a powerful property that can be used to establish the convergence  of learning models in various game-theoretic settings. For example, as shown in \cite{gao2020passivity}, exponential replicator dynamics \eqref{eq: EX-RD} is EI-passive, which ensures convergence  in contractive games.  In this section, we establish a connection between Finite regret property and EI-passivity.

Nash Stationarity is another fundamental property of the learning dynamic models that links the equilibria with the set of best responses to a deterministic payoff vector.
\begin{definition}[Nash Stationarity]
The learning dynamics specified by $\dot{x}= \mathcal{V}(x,p)$ satisfies Nash stationary if the following holds
\begin{equation*}
 \mathcal{V}(x^*,p^*)=0 \implies x^* \in \argmax_{z \in \Delta_n}  z^T p^* \; \; \forall p^*\in \mathbb{R}^n
\end{equation*}
\end{definition}

\begin{proposition} \label{prop: replicator EIP}
    The learning dynamic model is EI-passive if it has finite regret and satisfies Nash stationarity. 
\end{proposition}

\begin{proof}
Since the model has finite regret, there exists a constant $\beta \in \mathbb{R}$ such that 
\begin{equation*}
\langle p,x-\bar{x}\rangle_T 
      \geq \beta 
\end{equation*}
for all $T>0$ and every $\bar{x}\in \Delta_n$. To prove EI-passivity, we need to show that for every equilibrium $(p^*,x^*)$,  $\langle p-p^*,x-x^* \rangle_T $ is lower bounded for all $T>0$.  Choosing $\bar{x}= x^*$ in the finite regret inequality yields,  
\[
\langle p-p^*,x-x^* \rangle_T = \langle p,x-x^* \rangle_T - \langle p^*, x-x^* \rangle_T 
\]
By Nash stationarity, for any equilibrium $(p^*,x^*)$, we have  ${p^*}^T x(t) \leq {p^*}^T x^* \; \forall t>0$, so that  $\langle p^*, x-x^* \rangle_T \le 0 \; \forall T>0$. Thus, 
\begin{equation*}
     \langle p-p^*,x-x^* \rangle_T \geq \langle p,x-x^* \rangle_T \geq \beta
\end{equation*}
This lower bound establishes EI-passivity.
\end{proof}
In particular, replicator dynamics (RD) \eqref{eq: RD equation} satisfy Nash stationarity \cite{sandholm2009pairwise} and has finite regret, thus it is  EI-passive. Similarly, (DP) dynamics \eqref{eq:DPD} satisfies Nash stationarity \cite{lahkar2008projection} and, by Theorem~\ref{thrm: DP finite regret}, has finite regret; hence, it is EI-passive as well. Moreover, both SHO-FTRL dynamics \eqref{eq: SHO-FTRL} and SHO-DP dynamics \eqref{eq: SHO-DPD} are constructed as negative feedback interconnections between EI-passive learning dynamic models and an EI-passive linear system, which implies that they are EI-passive.
\begin{proposition}
    If a learning dynamic model is EI-passive and every $\bar{x} \in \Delta_n$ is an equilibrium point for the constant payoff vector $\mathbf{1}_n$, then the model has finite regret. 
\end{proposition}
\begin{proof}
    EI-passivity guarantees the existence of a  constant $\beta \in \R$ such that for any equilibrium pair $(p^*,x^*)$ the inequality  $\langle p-p^*,x-x^*\rangle_T \geq \beta$ holds for all $T>0$. In particular, since every $\bar{x} \in \Delta_n$ is an equilibrium when $p^*=\mathbf{1}_n$, we have, 
    \(\langle p-\mathbf{1}_n, x-\bar{x}\rangle_T\geq \beta \). 
Since  $x(t)-\bar{x}$ lies in $T\Delta_n(x(t))$,  $\langle \mathbf{1}_n, x(t)-\bar{x}\rangle = 0 $ for all $t \in \R_+$. Thus, 
\[\langle p-\mathbf{1}_n, x-\bar{x}\rangle_T = \langle p,x-\bar{x}\rangle_T \geq \beta \; \; \forall T>0 \]
and every $\bar{x} \in \Delta_n$. This lower bound implies that the model is passive from $p$ to $x-\bar{x}$ for every $\bar{x} \in \Delta_n$ and therefore has finite regret by theorem \ref{thrm:no_regret_theorem}.
\end{proof}
Exponential replicator dynamics (Ex-RD) \eqref{eq: EX-RD} satisfies the EI-passivity condition. However, for the  payoff vector  \(\mathbf{1}_n\), the unique equilibrium is the  strategy \(\frac{1}{n}\mathbf{1}_n\). As a result,  (Ex-RD) does not satisfy the condition that every \(\bar{x}\in\Delta_n\) can be an equilibrium for the constant payoff vecotor $\mathbf{1}_n$, and therefore it does not have finite regret. 

\begin{proposition} \label{prop: not finite regret not EI-passive}
If a learning dynamic model does not have finite regret and every \(\bar{x}\in\Delta_n\) is an equilibrium for the constant payoff vector \(\mathbf{1}_n\), then the model is not EI-passive.
\end{proposition}

\begin{proof}
Since the model does not have finite regret, there exists a strategy \(\bar{x}\in\Delta_n\) such that
\[
\limsup_{T\to\infty} -\langle p,x-\bar{x}\rangle_T= \infty.
\]
Given that every \(\bar{x}\in\Delta_n\) is an equilibrium for \(p^*=\mathbf{1}_n\), consider the equilibrium \((p^*,x^*)=(\mathbf{1}_n,\bar{x})\). Since \(x(t)-\bar{x}\) lies in $T\Delta_n(x(t))$, 
\(
\mathbf{1}_n^\top\bigl(x(t)-\bar{x}\bigr)=0.
\)
It follows that
\[
\limsup_{T\to\infty} -\langle p-\mathbf{1}_n,x-\bar{x}\rangle_T = \limsup_{T\to\infty} -\langle p,x-\bar{x}\rangle_T =\infty.
\]
which violates the EI-passivity condition. Therefore, the model is not EI-passive.
\end{proof}
Example \ref{ex: no regret counter example} shows that BNN dynamics \eqref{eq: BNN dynamics}, Smith dynamics \eqref{eq: Smith dynamics}, Logit dynamics \eqref{eq: logit dynamics}, and (TP) dynamics \eqref{eq: target projection dynamics} do not finite regret. Moreover, For each of these  models, every $\bar{x} \in \Delta_n$ is an equilibrium point for the constant payoff vector $\mathbf{1}_n$. Consequently, Proposition \ref{prop: not finite regret not EI-passive} implies that none of these learning models is EI-passive.

\section{Passivity-Based classification} \label{sec: Passivity-based classification}
In this section, we demonstrate that incremental passivity is a generalized passivity notion that implies both \(\delta\)-passivity and EI-passivity. Accordingly, we introduce a classification of learning dynamic models based on the passivity notion they satisfy (i.e., EI-passivity, \(\delta\)-passivity, and incremental passivity). In Section \ref{input-output operator passivity}, we showed that incremental passivity implies EI-passivity in both input–output operators and state-space formulations. The following propositions establish that incremental passivity further implies \(\delta\)-passivity.
\begin{proposition} \label{prop: incremental passivity to delta passivity input -out operator} 
     An input-output operator $H$ is  $\delta$-passive if it incrementally passive. 
\end{proposition}
\begin{proof}
    Let  $(u,y)$ are the input and the output of  $H$. Take $(\tilde{u},\tilde{y})$ as the time shifted version of $(u,y)$ by small $\tau>0$. by time-invariance of $H$, $\tilde{u}=u(t-\tau)$ and $\tilde{y}=y(t-\tau)$. The Taylor series expansions of $\tilde{u}$ and $\tilde{y}$ are 
\begin{equation}
    \begin{gathered}
        u(t-\tau) = u(t)- \tau \dot{u} + O(\tau^2)\\
        y(t-\tau) = y(t)- \tau \dot{y} + O(\tau^2) 
    \end{gathered}
    \label{eq: shifted trajectroy}
\end{equation}
From incremental passivity, we have for all $T \in \mathbb{R}_+$ 
\begin{equation*}
    \langle u-\tilde{u},y-\tilde{y}\rangle_T \geq 0 
\end{equation*}
substitute by \eqref{eq: shifted trajectroy} yields
\begin{equation*}
    \int_0^T (\tau \dot{u} + O(\tau^2))^T (\tau \dot{y}+ O(\tau^2))dt = \int_0^T \tau^2 \dot{u}^T \dot{y} + O(\tau^3) dt\geq 0 
\end{equation*}
Divide by $\tau^2$ and take the limit as $\tau \to 0$, we get 
\begin{equation*}
    \int_0^T \dot{u}^T \dot{y} dt =\langle \dot{u},\dot{y}\rangle_T\geq0 
\end{equation*}
which is the condition for \(\delta\)-passivity with a lower bound $\beta=0$
\end{proof}
\begin{proposition}\label{prop: incremntal passivity to delta passivity state space model}
    The state space model \eqref{eq: state space model} is $\delta$-passive if it is incrementally passive. 
\end{proposition}
\begin{proof}
Let \(V_\Delta(x,\tilde{x})\) be an incremental storage function such that for any two trajectories \((u,x,y)\) and \((\tilde{u},\tilde{x},\tilde{y})\),
\begin{equation*}
    \dot{V}_\Delta(x(t),\tilde{x}(t)) \leq (u(t)-\tilde{u}(t))^T (y(t)-\tilde{y}(t)) 
\end{equation*}
    Let  $(\tilde{u},\tilde{x},\tilde{y})$ is a time shifted trajectory of  $(u,x,y)$ by small $\tau>0$. By time invariance,  $(\tilde{u},\tilde{x},\tilde{y})=(u(t-\tau),x(t-\tau),y(t-\tau))$. Hence,   
\begin{equation*}
\begin{aligned}
    \dot{V}_{\Delta}(x(t),x(t-\tau)) &\leq (u(t)-u(t-\tau))^T(y(t)-y(t-\tau)) \\ 
    &=\tau^2 \dot{u}^T\dot{y} + O(\tau^3)
\end{aligned}
\end{equation*}
Dividing both sides by $\tau^2$ and taking the limit as $\tau \to 0$ yields, 
\begin{equation*}
    \lim_{\tau \to 0} \frac{1}{\tau^2}\dot{V}_{\Delta}(x(t),x(t-\tau)) \leq \dot{u}^T \dot{y}
\end{equation*}
Therefore, the $\delta$-passivity storage function is 
\begin{equation*}
    V_\delta(x(t),\dot{x}(t)) = \lim_{\tau \to 0} \frac{1}{\tau^2} V_{\Delta}(x(t),x(t-\tau)).
\end{equation*}

To show that $V_{\delta}(x(t),\dot{x}(t))$ is positive semi-definite function, we write the Taylor series expansion of $V_{\Delta}(x,y)=V_{\Delta}(x(t),x(t-\tau))=V_{\Delta}(x(t),x(t)+\Delta x_2)$ around the second argument $y$ where $\Delta x_2=-\tau \dot{x}+O(\tau^2)$, we get 
\begin{equation*}
\begin{aligned}
        V_{\Delta}(x(t),x(t-\tau))
        &= V_{\Delta}(x(t),x(t))+\nabla_{y} V_{\Delta}(x(t),x(t)) \Delta x_2\\
        &+ \frac{1}{2} \Delta x_2^T \nabla^2_{y} V_{\Delta}(x(t),x(t))\Delta x_2+ H.O.T 
\end{aligned}
\end{equation*}
The first term  $V_{\Delta}(x(t),x(t))=0$ and since $V_{\Delta}$ is continuously differentiable positive semi-definite and has a local minimum at $V_{\Delta}(x(t),x(t))$, then $\nabla_{y}V_{\Delta}(x(t),x(t))=0$.   Therefore, 
\begin{equation*}
\begin{aligned}
     V_{\Delta}(x(t),x(t-\tau))
     &=\frac{1}{2} \tau^2 \dot{x}(t) \nabla^2_{y} V_{\Delta}(x(t),x(t)) \dot{x}(t) + O(\tau^3)
\end{aligned}
\end{equation*}
By dividing by $\tau^2$ and taking the limit as $\tau \to 0$, we obtain 
\begin{equation*}
     V_\delta(x(t),\dot{x}(t)) = \frac{1}{2} \dot{x}(t) \nabla^2_{y} V_{\Delta}(x(t),x(t)) \dot{x}(t) \geq0
\end{equation*}
since the Hessian \(\nabla^2_{y} V_\Delta(x,y)\) is positive semi-definite. This completes the proof.
 
\end{proof}
Prior work has shown that BNN dynamics \eqref{eq: BNN dynamics} and Smith dynamics \eqref{eq: Smith dynamics} are \(\delta\)-passive \cite{fox2013population}. Additionally,   \cite{park2018passivity} established  that Logit dynamics \eqref{eq: logit dynamics} is  \(\delta\)-passive, while  replicator dynamics (RD) \eqref{eq: RD equation} is not $\delta$-passive. In this work, we extend this analysis by showing that target projection (TP) dynamics \eqref{eq: target projection dynamics} is $\delta$-passive. 
\begin{proposition}
    Target projection (TP) dynamics \eqref{eq: target projection dynamics} is $\delta$-passive. 
\end{proposition}
\begin{proof}
    To establish \(\delta\)-passivity for a learning model of the form \(\dot{x}=\mathcal{V}(x,p)\), we seek a positive semi-definite storage function \(V(x,p)\) such that
\(
\dot{V}(x,p) \le \dot{x}^\top \dot{p}.
\)
For (TP) dynamics, define
\[
V(x,p)=\max_{y\in\Delta_n}(y-x)^\top p - \frac{1}{2}\|y-x\|_2^2.
\]
      and Let 
    \(
         y^*(x,p) = \argmax_{y\in \Delta_n} (y-x)^T p -\frac{1}{2} \|y-x\|_2^2
    \), then, 
     \begin{align*}
        y^*(x,p)&=\argmin_{y \in \Delta_n} \|y-x\|_2^2 + \|p\|_2^2 -2 (y-x)^T p \\
        &= \argmin_{y\in \Delta_n} \|y- (x+p)\|_2^2  = \Pi_\Delta(x+p)
    \end{align*}
Thus, we can write
\begin{equation*}
    V(x,p) = (y^*(x,p)-x)^T p- \frac{1}{2} \|y^*(x,p)-x\|_2^2 
\end{equation*}
By  envelop theorem, the gradient of $V(x,p)$ w.r.t. $x$ and $p$ are 
\begin{align*}
    \nabla_x V(x,p)&= -p + \Pi_\Delta(x+p)-x \\
    \nabla_p V(x,p) & = \Pi_\Delta(x+p)-x = \dot{x} 
\end{align*}
Hence, 
\begin{align*}
    \dot{V}(x,p)& = \nabla_x V(x,p) \dot{x} + \nabla_p V(x,p)\dot{p} \\ 
    &= ((x+p)-\Pi_\Delta(x+p))^T (x-\Pi_\Delta(x+p)) + \dot{x}^T \dot{p}\\
    &\leq \dot{x}^T \dot{p} 
\end{align*}
where the final inequality follows from the properties of the projection operator. This confirms that TP dynamics are \(\delta\)-passive.

\end{proof}
In Section \ref{sec: finite regret and EI-passivity}, we showed that replicator dynamics (RD) \eqref{eq: RD equation}, DP dynamics \eqref{eq:DPD}, (SHO-FTRL) dynamics \eqref{eq: SHO-FTRL}, and (SHO-DP) dynamics \eqref{eq: SHO-DPD} have finite regret and are EI-passive, whereas BNN dynamics \eqref{eq: BNN dynamics}, Smith dynamics \eqref{eq: Smith dynamics}, Logit dynamics \eqref{eq: logit dynamics}, and (TP) dynamics \eqref{eq: target projection dynamics} fail to achieve finite regret and are not EI-passive. Moreover, while exponential replicator dynamics is EI-passive, it does not have finite regret. Together with the \(\delta\)-passivity results, these findings form the basis of our passivity-based classification, as illustrated in Figure \ref{fig: passivity based classification}.

\begin{figure}[H]
    \centering
    \includegraphics[width=0.8\linewidth]{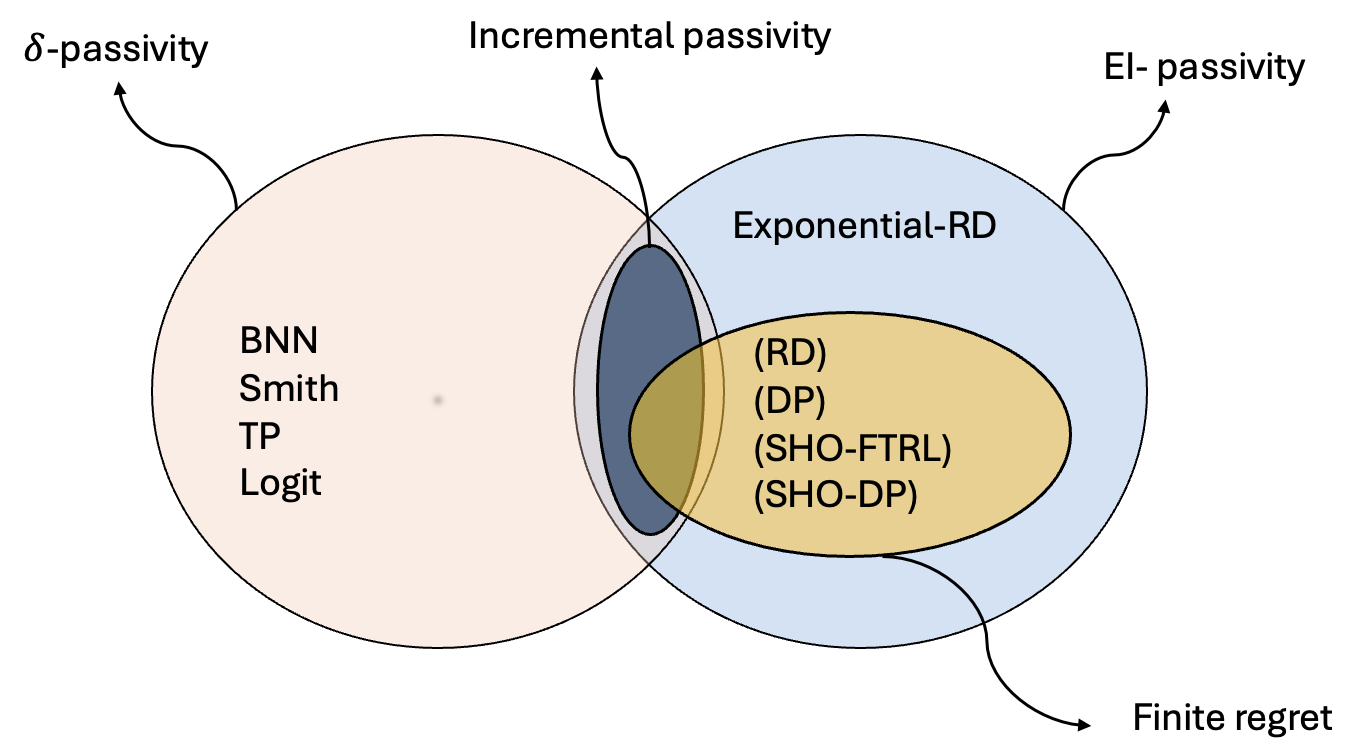}
    \caption{Passivity-based classification of learning dynamic models}
    \label{fig: passivity based classification}
\end{figure}

\section{Convergence in Contractive games} \label{sec: Convergence in contracive games}
The implementation of a learning dynamic model in a game can be viewed as a feedback interconnection between the dynamic system (i.e., the learning model) and the game itself, which may be static or dynamic (see Figure~\ref{fig:connecting LDM with a game}). In this work, we focus on static contractive games where the payoff is given by \(p=F(x)\).
\begin{figure}[H]
    \centering
    \includegraphics[width=0.5\linewidth]{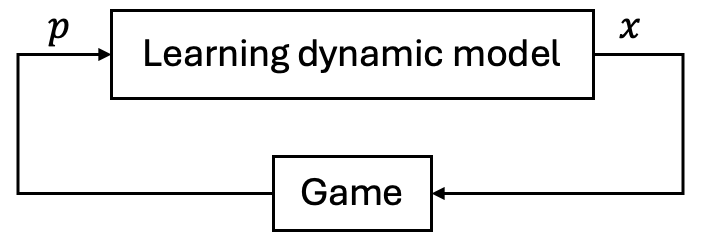}
    \caption{Feedback interconnection between a learning dynamic model and a game}
    \label{fig:connecting LDM with a game}
\end{figure}
 
We now show that if a learning dynamic model has finite regret, then it converges globally in strictly contractive games to the unique Nash equilibrium.

\begin{theorem}
    If the learning dynamic model has finite regret, then it converges globally in strictly contractive games to the unique Nash equilibrium. 
\end{theorem}

\begin{proof}
If the learning dynamic model has finite regret, it is passive from \(p=F(x)\) to \(x-\bar{x}\) for every \(\bar{x}\in\Delta_n\). Consequently, there exists a storage function \(V(x)\ge 0\) with \(V(\bar{x})=0\) such that 
\(
\dot{V}(x) \le F(x)^\top (x-\bar{x}).
\)
In a strictly contractive game, the Nash equilibrium \(x^*\) is unique; choosing \(\bar{x}=x^*\) and using strict contractiveness, we get 
\[
\dot{V}(x) \le F(x)^\top (x-x^*) < F(x^*)^\top (x-x^*) \le 0,
\]
where the final inequality follows from the definition of Nash equilibrium. Since \(\dot{V}(x)<0\) for all \(x\ne x^*\) and \(V(x^*)=0\), \(V(x)\) is a Lyapunov function that guarantees global convergence to \(x^*\).
\end{proof}

By Definition~\ref{def: contracrive games}, contractive games are anti-incrementally passive. Hence, they are also anti-EI-passive and anti-\(\delta\)-passive (see Proposition~\ref{prop: incremental passivity to delta passivity input -out operator}). Consequently, when a contractive game is interconnected with a learning dynamic model that is incremental passive, \(\delta\)-passive, or EI-passive, the closed-loop system has stable behavior. The passivity-based classification of learning dynamic models (see Figure~\ref{fig: passivity based classification}) indicates that any model within this classification map is stable in contractive games and converges globally to the unique Nash equilibrium in strictly contractive games.

\section{CONCLUSIONS}
In this paper, we established  a connection between passivity, no-regret, and convergence in contractive games for various learning dynamic models. We showed that when a learning algorithm satisfies passivity  from the payoff vector to the deviation between its output strategy and any fixed strategy, it guarantees finite regret. This property was illustrated for the continuous-time versions of (FTRL) dynamics and (DP) dynamics as well as for their strategic higher-order variants, (SHO-FTRL) and (SHO-DP). Our numerical experiments further show that several evolutionary dynamics, including BNN, Smith, Logit, and TD dynamics, do not have finite regret.  We also examined the fragility of this property under delayed or perturbed payoffs. Moreover, we introduced a passivity-based classification of learning dynamics, based on incremental passivity, \(\delta\)-passivity, and EI-passivity, which provides a framework for analyzing convergence in contractive games.

In future work, we aim to study the finite regret properties of the payoff-based higher-order variants of learning dynamics that already have finite regret. Moreover, since many standard evolutionary dynamic models fail to guarantee finite regret, we plan to enhance these models by designing higher-order variants that ensure finite regret and preserve their desired convergence properties in contractive games.

\addtolength{\textheight}{-12cm}   

\bibliographystyle{ieeetr}

\bibliography{references.bib}

\end{document}